\documentclass[journal,twoside,web]{ieeecolor}

\usepackage{generic}
\usepackage{cite}
\usepackage{amsmath,amssymb,amsfonts,amsthm}
\usepackage{algorithmic}
\usepackage{graphicx}
\usepackage{textcomp}

\usepackage{fontenc}

\usepackage{scalefnt}

\makeatletter
\def\fnum@figure{\textcolor{subsectioncolor}{\sf Fig.~\thefigure}}
\def\fnum@table{\textcolor{subsectioncolor}{\sf TABLE~\thetable}}
\makeatother
\newtheorem{theorem}{Theorem}
\newtheorem{definition}{Definition}
\newtheorem{lemma}{Lemma}
\newtheorem{remark}{Remark}

\DeclareMathOperator*{\argmin}{arg\,min}

\def\BibTeX{{\rm B\kern-.05em{\sc i\kern-.025em b}\kern-.08em
    T\kern-.1667em\lower.7ex\hbox{E}\kern-.125emX}}
\begin{document}
\title{Generalization of Safe Optimal Control Actions on Networked Multi-Agent Systems}
\author{Lin Song, Neng Wan, Aditya Gahlawat, Chuyuan Tao, Naira Hovakimyan, and Evangelos A. Theodorou
\thanks{This work is  supported by the Air Force Office of Scientific Research (AFSOR), National Aeronautics and Space Administration (NASA) and National Science Foundation's National Robotics Initiative (NRI) and Cyber-Physical Systems (CPS) awards \#1830639, \#1932529, and \#1932288. }
\thanks{Lin Song, Neng Wan, Aditya Gahlawat, Chuyuan Tao, and Naira Hovakimyan are with the Department of Mechanical Science and Engineering, University of Illinois at Urbana Champaign, Urbana, IL 61801 USA (e-mail:\{linsong2, nengwan2, gahlawat, chuyuan2, nhovakim\}@illinois.edu).}
\thanks{Evangelos A. Theodorou is with the Department of Aerospace Engineering, Georgia Institute of Technology, Atlanta, GA 30332 USA (e-mail: \{evangelos.theodorou\}@gatech.edu).}
}

\maketitle

\begin{abstract}
We propose a unified framework to fast generate a safe optimal control action for a new task from existing controllers on Multi-Agent Systems (MASs). The control action composition is achieved by taking a weighted mixture of the existing controllers according to the contribution of each component task. Instead of sophisticatedly tuning the cost parameters and other hyper-parameters for safe and reliable behavior in the optimal control framework,  the safety of each single task solution is guaranteed using the control barrier functions (CBFs) for high-degree stochastic systems, which constrains the system state within a known safe operation region where it originates from. Linearity of CBF constraints in control enables the control action composition. The discussed framework can immediately provide reliable solutions to new tasks by taking a weighted mixture of solved component-task actions and filtering on some CBF constraints, instead of performing an extensive sampling to achieve a new controller. Our results are verified and demonstrated on both a single UAV and two cooperative UAV teams in an environment with obstacles.
\end{abstract}


\section{Introduction}
\label{sec:introduction}
Optimal  control design by minimizing a cost function or maximizing a reward function has been investigated for various systems \cite{kirk2004optimal,sutton2018reinforcement}. Specifically, stochastic optimal control problems consider minimizing a cost function for dynamical systems subject to random noise \cite{kumar2015stochastic,kappen2005linear}. Although such formulation captures a wide class of real-world problems, stochastic optimal control actions are typically difficult and expensive to compute in large-scale systems due to the curse of dimensionality \cite{blondel2000survey}. To overcome the computation challenges, many approximation-based approaches, including cost parameterization\cite{bertsekas1995neuro}, path-integral formulation \cite{theodorou2010generalized,theodorou2010reinforcement}, value function approximation \cite{powell2011review} and policy approximation \cite{sutton2000policy}, have been proposed. Exponential transformation on the value function was applied in \cite{todorov2009efficient} 
such that linear-form solutions to stochastic control problems were achieved, and thus the computational efficiency was improved. Optimal control problems whose solutions can be obtained by solving reduced linear equations are generally categorized as linearly-solvable optimal control (LSOC) problems in \cite{dvijotham2012linearly}. These formulations leverage the benefits of LSOC problems, including compositionality and  path-integral representation of the optimal solution \cite{dvijotham2011unified}. Path integral (PI) control approach usually requires extensive sampling on given dynamical models. However, the compositionality property enables the construction of composite control by taking a weighted mixture of existing component controllers to solve a new task in a certain class without sampling on the dynamical system again, and thus improves the computation efficiency \cite{pan2015sample,todorov2009compositionality}. The methodology of combining and reusing existing controllers has been explored and validated on different systems, from single-agent systems (e.g. 3-dimensional robotic arm \cite{pan2015sample} and physically-based character animation \cite{muico2011composite}) to networked multi-agent systems (e.g. a cooperative UAV team \cite{song2021compositionality}).  

Networked multi-agent systems enable coordination between agents and have a wide range of real-world applications, including cooperative vehicles \cite{mahony2012multirotor,cichella2016safe}, robotics \cite{wilson2020robotarium}, sensor networks \cite{zheng2018average}, and transportation \cite{liu2017intelligent}. The decentralized POMDP (Dec-POMDP) model in \cite{bernstein2002complexity}, formulates the control problem in multi-agent systems under uncertainty with only local state information available. A dynamic programming (DP)-based approach has been investigated on the Dec-POMDP model in \cite{hansen2004dynamic} towards the search for the optimal solution. However, the computation required for solving such systems grows exponentially as the system state dimension scales. To tackle the difficulties in computation, approximation-based approaches and distributed algorithms based on certain network partition have been discussed. In \cite{pajarinen2011periodic}, the expectation-maximization (EM) algorithm was introduced on the Dec-POMDP model, where optimal policies were represented as finite-state controllers to improve the scalability. Q-value functions are approximated and efficient computation of optimal policies is achieved in \cite{oliehoek2008optimal}.  Furthermore, fully-decentralized reinforcement learning algorithms were investigated on networked MASs using function approximation in \cite{zhang2018fully}. Compared to aforementioned approaches, path integral (PI) formalism of stochastic optimal control problem solutions relies on sampling of the stochastic differential equation, and is applicable to large-scale nonlinear systems. Path integral formulation has demonstrated higher efficiency and robustness in solving high-dimensional reinforcement learning problems in \cite{theodorou2010generalized}. Path integral control approach has also been extended to MASs, and an approximate inference method was applied to approximate the exact solution in \cite{van2008graphical}. A distributed algorithm proposed in \cite{wan2021distributed} partitions the networked MAS into multiple subsystems, from which local control actions are computed with limited resources.

However, the optimal solution to a cost-minimization problem is not always reliable and applicable in the real world, especially in safety-critical scenarios. Guaranteed-safe control actions have attracted researchers' interest, both in  controller design \cite{aswani2013provably,holicki2020controller}, and in  verification\cite{tabuada2009verification, mitra2013verifying}.  Control barrier functions (CBFs) have been proposed as a powerful tool to enforce  safe behavior of system states by introducing extra constraints on the control inputs \cite{ames2016control}.  Control barrier functions  are required to be finite (for Reciprocal CBF) or positive (for Zero CBF) as the system operates within the known safe region. CBF can be constructed empirically \cite{ames2019control} or learned from data \cite{srinivasan2020synthesis,robey2020learning}. Along with optimization-based control actions or reinforcement learning (RL) techniques, the CBF-based approach to ensure safety has achieved satisfying results in various application scenarios, including bipedal locomotion under model uncertainty \cite{choi2020reinforcement}, autonomous vehicles \cite{ames2014control}, and UAVs \cite{chen2020guaranteed}. Recently, CBF techniques were generalized to stochastic systems with high-probability guarantees, in cases of both complete and incomplete information in \cite{clark2021control} and \cite{clark2019control}. To achieve the state-trajectory tracking goal and associated safe certificates on nonlinear and complicated systems, contraction-based methods are applicable  on both nominal dynamics \cite{manchester2017control} and uncertain dynamics \cite{gahlawat2021contraction,lakshmanan2020safe}. 

\begin{figure}[!t]
\centerline{\includegraphics[width=\columnwidth]{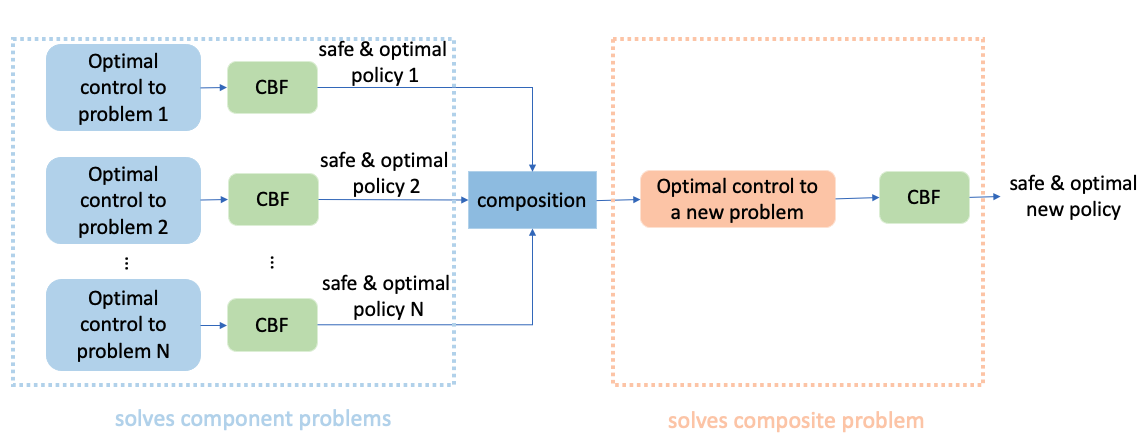}}
\caption{The architecture of the proposed safe composition control framework.}
\label{fig_structure}
\end{figure}
However, efficient computation of certified-safe stochastic optimal control solutions on MASs is still an open problem. The main contribution of this paper is a framework of generalizing optimal control actions while ensuring safety on networked MASs; the architecture is illustrated in Fig. ~\ref{fig_structure}. When multiple solved control problems on MASs share identical dynamical information, but are slightly different in some aspects, such as terminal states and final costs, instead of sampling on the given dynamical system again to solve a new task, the compositionality of achieved control actions can be leveraged, and the existing controllers of these solved problems can be weighted and mixed to solve a new problem and drive the system to a new target.  The baseline optimal control for the new task is first obtained by mixing existing certified-safe controllers.  Then, a post-composite constrained optimization using CBFs is formulated to filter the baseline control and thus guarantee the safety. 
The task generalization capability of resulting control actions allows direct and efficient computation of a new problem solution without re-sampling in a certain class. Compared with our preliminary work \cite{song2021compositionality}, we furthermore incorporate safety constraints on the compositionality of LSOC problem solutions. The proposed strategy is validated via numerical simulations on both single UAV and two cooperative UAV teams. 

The rest of the paper is organized as follows: Section \ref{sec2} introduces the preliminaries of formulating stochastic control problems, linearly-solvable optimal control (LSOC) problems, and stochastic CBFs; Section \ref{sec3} introduces the compositionality of LSOC problem solutions on networked MASs; Section \ref{sec4} introduces the achieved certified-safe optimal control actions on MASs and the generalization of proposed control actions with safety guarantees; Section \ref{sec5} provides numerical simulations in three scenarios validating the proposed approach; the conclusion of this paper and some open problems are discussed in Section \ref{sec6}.

\section{Preliminaries and Problem Formulation}\label{sec2}

We start by introducing some preliminary results of optimal control actions in stochastic systems, including both the single-agent and multi-agent scenarios. Then an extension of control barrier functions (CBFs) to stochastic systems is presented, which enables the proposed safe and optimal control framework introduced in Section \ref{sec4}.

\subsection{Stochastic optimal control problems}\label{pf:soc-problems}
\subsubsection{Single-Agent Systems}
Consider a continuous-time dynamical system described by the It$\hat{\textrm{o}}$ diffusion process: \begin{align}
    dx^t &= g(x^t)dt+B(x^t)[u(x^t,t)dt+\sigma d\omega] \label{eq:cont_dyn}\\
    &= f(x^t,u^t)dt+F(x^t,u^t)d\omega,
\nonumber
\end{align}
where time-varying $x^t \in \mathbb{R}^M$ is the state vector with $M$ denoting the state dimension, $g(x^t) \in \mathbb{R}^{M}$, $B(x^t) \in \mathbb{R}^{M \times P}$,  $u(x^t,t) \in \mathbb{R}^{P}$ are the passive dynamics, control matrix and control input with $P$ denoting the input dimension, and $\omega \in \mathbb{R}^P$ denotes the Brownian noise with covariance $\sigma \in \mathbb{R}^{P \times P}$.

Let $\mathcal{I}$ denote the set of interior states and $\mathcal{B}$ denote the set of boundary states. When $x^t \in \mathcal{I}$, the running cost function is defined as:\begin{equation}\label{eq:sas-running-cost}
    c(x^t,u^t) = q(x^t) + \frac{1}{2}u(x^t,t)^\top Ru(x^t,t),
\end{equation}
where $q(x^t) \ge 0$ is a state-related cost and $u(x^t,t)^\top Ru(x^t,t)$ is a control-quadratic term with $R$ being a positive definite matrix. When $x^{t_f} \in \mathcal{B}$, the terminal cost function is denoted by $\phi(x^{t_f})$, where $t_f$ is the exit time and determined online in the first-exit problem. The cost-to-go function $J^{u}(x^t,t)$ for the first-exit problem under control action $u$ can be defined as:\begin{equation}\label{eq:sas-cost-to-go}
    J^u(x^t,t) = \mathbb{E}_{x^t,t}^{u}[\phi(x^{t_f})+\int_t^{t_f}c(x(\tau),u(\tau))d\tau].
\end{equation}
The value function $V(x^t,t)$ is defined as the optimal cost-to-go function:\begin{equation}\label{eq:sas-value-func}
    V(x^t,t) = \min_u \mathbb{E}_{x^t,t}^{u}[\phi(x^{t_f})+\int_t^{t_f}c(x(\tau),u(\tau))d\tau].
\end{equation}
The value function is the expected cumulative running cost starting from the state $x^t$ and acting optimally thereafter. For notation simplicity, the time-evolution of the state is omitted.

Define a stochastic second-order differentiator as $\mathcal{L}_{(u)}[V] = f^\top \nabla_{x}V+\frac{1}{2}\textrm{tr}(FF^\top\nabla_{xx}^2V)$, and then the value function $V$ satisfies the stochastic Hamilton-Jacobi-Bellman (HJB) equation as follows:\begin{equation}\label{eq:sas-shjb}
    0 = \min_u\{c(x,u)+\mathcal{L}_{(u)}[V](x)\}.
\end{equation}
With the desirability function $Z(x,t) = [\exp(-V(x,t))/\lambda]$ and under the nonlinearity cancellation condition $\sigma\sigma^\top = \lambda R^{-1}$, the linear-form optimal control action for continuous-time stochastic systems  is obtained as \begin{equation}\label{eq:sas-optimal-control-pf}
    u^*(x,t) = \sigma \sigma^\top B^\top(x) \frac{\nabla_x Z(x,t)}{Z(x,t)},
\end{equation}
and the corresponding transformed linear HJB equation takes the form of $0 = \mathcal{L}[Z]-qZ$, where $\mathcal{L}[Z] = f^\top \nabla_{x}Z+\frac{1}{2}\textrm{tr}(FF^\top\nabla_{xx}^2Z)$.

\subsubsection{Multi-Agent Systems}\label{pf:mas-setup} For a networked multi-agent system governed by mutually independent passive dynamics, the index set of all the agents neighboring or adjacent to agent $i$ is denoted by $\mathcal{N}_i$. The factorial subsystem for agent $i$ includes all the agents directly communicating with agent $i$ and the agent itself, and is denoted by $\bar{\mathcal{N}}_i := \mathcal{N}_i \cup \{i\}$, and the cardinality of set $\bar{\mathcal{N}}_i$ is denoted by $|\bar{\mathcal{N}}_i|$. Consider the joint continuous-time dynamics  for factorial subsystem  $\bar{\mathcal{N}}_i$ as in \cite{wan2021distributed}:
\begin{equation}\label{eq:cont_dyn_mas}
    d\bar{x}_i = \bar{g}_i(\bar{x}_i)dt+\bar{B}_i(\bar{x}_i)[\bar{u}_i(\bar{x}_i,t)dt+\bar{\sigma}_i d\bar{\omega}_i], 
\end{equation}
where the joint state vector is denoted by $\bar{x}_i = [x_i^\top,x_{j\in \mathcal{N}_i}^\top]^\top \in \mathbb{R}^{M \cdot |\bar{\mathcal{N}}_i|}$, the joint passive dynamics vector is denoted by $\bar{g}_i(\bar{x}_i) = [g_i(x_i)^\top,g_{j\in \mathcal{N}_i}(x_j)^\top]^\top \in \mathbb{R}^{M\cdot|\bar{\mathcal{N}}_i|}$, the joint control matrix is denoted by $\bar{B}_i(\bar{x}_i) = [B_i(x_i) \hspace{1mm}  \mathbf{0};\mathbf{0} \hspace{1mm} B_{j \in \mathcal{N}_i}(x_j)] \in \mathbb{R}^{M \cdot |\bar{\mathcal{N}}_i|\times P \cdot |\bar{\mathcal{N}}_i|}$, the joint control action vector is $\bar{u}_i(\bar{x}_i,t) = [u_i(\bar{x}_i,t)^\top,u_{j\in \mathcal{N}_i}(\bar{x}_i,t)^\top]^\top \in \mathbb{R}^{P\cdot |\bar{\mathcal{N}}_i|}$, and $\bar{\omega}_i = [\omega_i^\top,\omega_{j \in \mathcal{N}_i}^\top]^\top \in \mathbb{R}^{P \cdot |\bar{\mathcal{N}}_i|}$ is the joint noise vector with covariance matrix $\bar{\sigma}_i = \textrm{diag}\{\sigma_i,\sigma_{j\in \mathcal{N}_i}\} \in \mathbb{R}^{P \cdot |\bar{\mathcal{N}}_i| \times P \cdot |\bar{\mathcal{N}}_i|}$.

Let $\bar{\mathcal{I}}_i$ denote the set of interior joint states, and $\bar{\mathcal{B}}_i$ denote the set of boundary joint states. When $\bar{x}_i \in \bar{\mathcal{I}}_i$, the running cost function of the joint states of $\bar{\mathcal{N}}_i$ is defined as:\begin{equation}\label{eq:mas-running-cost}
    c_i(\bar{x}_i,\bar{u}_i) = q_i(\bar{x}_i) + \frac{1}{2}\bar{u}_i(\bar{x}_i,t)^\top \bar{R}_i\bar{u}_i(\bar{x}_i,t),
\end{equation}
where $q_i(\bar{x}_i) \ge 0$ is a joint-state-related cost, and $\bar{u}_i(\bar{x}_i,t)^\top \bar{R}_i\bar{u}_i(\bar{x}_i,t)$ is a control-quadratic term with $\bar{R}_i \in \mathbb{R}^{P \cdot |\bar{\mathcal{N}}_i| \times P \cdot |\bar{\mathcal{N}}_i|}$ being a positive definite matrix. When $\bar{x}_i^{t_f} \in \bar{\mathcal{B}}_i$, the terminal cost function is denoted by $\phi_i(\bar{x}_i^{t_f})$, where $t_f$ is the exit time and determined online in the first-exit problem. Specifically, when the central agent state $x_i^{t_f} \in \mathcal{B}_i$, we also have the terminal cost function $\phi_i(x_i^{t_f})$ defined. As an extension of the single-agent case, the cost-to-go function $J_i^{\bar{u}_i}(\bar{x}_i^t,t)$ for the first-exit problem under joint control action $\bar{u}_i$ can be defined as:\begin{equation}\label{eq:mas-cost-to-go}
    J_i^{\bar{u}_i}(\bar{x}_i^t,t) = \mathbb{E}_{\bar{x}_i^t,t}^{\bar{u}_i}[\phi_i(\bar{x}_i^{t_f})+\int_t^{t_f}c_i(\bar{x}_i(\tau),\bar{u}_i(\tau))d\tau].
\end{equation}
The value function $V_i(\bar{x}_i,t)$ is defined as the optimal cost-to-go function:\begin{equation}\label{eq:mas-value-func}
    V_i(\bar{x}_i,t) = \min_{\bar{u}_i} \mathbb{E}_{\bar{x}_i^t,t}^{\bar{u}_i}[\phi_i(\bar{x}_i^{t_f})+\int_t^{t_f}c_i(\bar{x}_i(\tau),\bar{u}_i(\tau))d\tau],
\end{equation}
and the value function is the expected cumulative running cost starting from joint state $\bar{x}_i$ and acting optimally thereafter. 

Similar to the single-agent scenario, the desirability function over the joint state $\bar{x}_i$ is defined as \begin{equation}\label{eq:mas-desir}
    Z(\bar{x}_i,t) = \exp[-V_i(\bar{x}_i,t)/\lambda_i]
\end{equation}   under the condition that $\bar{\sigma}_i\bar{\sigma}_i^\top = \lambda_i\bar{R}_i^{-1}$ is satisfied to cancel the nonlinear terms. The linear-form joint optimal control action for continuous-time stochastic networked MAS under aforementioned decentralization topology is derived in \cite{wan2021distributed}  and the result is in the form:\begin{equation}\label{eq:mas-optimal-control-pf}
    \bar{u}_i^*(\bar{x}_i,t) = \bar{\sigma}_i\bar{\sigma}_i^\top \bar{B}_i(\bar{x}_i)^\top\frac{\nabla_{\bar{x}_i} Z(\bar{x}_i,t)}{Z(\bar{x}_i,t)}. 
\end{equation}
\begin{remark}
Only the central agent $i$ of the factorial subsystem $\bar{\mathcal{N}}_i$ samples or selects the local optimal control action $u_i^*(\bar{x}_i,t)$ from the joint optimal control action $\bar{u}^*_i(\bar{x}_i,t)$ according to the computation results on factorial subsystem $\bar{\mathcal{N}}_i$. 
\end{remark}

\subsection{Stochastic control barrier function (SCBF) as a safety filter}
Optimization-based control actions achieve the goal of cost minimization, but seldom provide rigorous guarantees on satisfying the commanded constraints, which can be problematic in some safety-critical scenarios. Control barrier functions (CBFs) have been widely used to describe certain desired  known safe operating region and can guarantee the state-invariant property within the proposed set by filtering existing control actions. Recently, CBFs have also been investigated in stochastic systems in \cite{clark2021control}.  
\begin{definition}[zero-CBF for stochastic systems]
The function $h(x)$ serves as a zero-CBF (ZCBF) for a system described by the stochastic differential equation (SDE) \eqref{eq:cont_dyn}, if for all $x$ satisfying $h(x) > 0$ there exists a $u$ satisfying \begin{equation}\label{eq:zcbf-def}
    \frac{\partial{h}}{\partial x}(g(x)+B(x)u)+\frac{1}{2}\textrm{tr}(\sigma^\top B^\top(x)(\frac{\partial^2 h}{\partial x^2})  B(x)\sigma) \ge -h(x).
\end{equation}
\end{definition}

We next introduce a lemma and a theorem modified from \cite{clark2021control} constructing ZCBFs for high-degree stochastic systems, which our proposed safe control framework is built upon.
\begin{lemma}[State-invariant guarantees by definition of the ZCBF]\label{lemma:state-invariant-zcbf}
Define the safe operating region $\mathcal{C}=\{x:h(x) \ge 0\}$. If $u^t$ satisfies \eqref{eq:zcbf-def} for all $t$, then $P_r(x^t \in \mathcal{C} \hspace{2mm} \forall t)=1$, provided $x^0 \in \mathcal{C}$.
\end{lemma}

\begin{theorem}[Construction of ZCBF  for  high-degree stochastic systems]
\label{lemma:zcbf-high-deg}
Define the ZCBFs for high-degree stochastic systems \eqref{eq:cont_dyn} as $h_0(x) = h(x)$, and subsequently \begin{align}\label{eq:construct-zcbf-high-deg}
    h_{i+1}(x) &= \frac{\partial{h_i}}{\partial x}g(x)\nonumber\\ &+ \frac{1}{2}\textrm{tr}(\sigma^\top B^\top(x)(\frac{\partial^2 h_i}{\partial x^2})  B(x)\sigma)+h_i(x),
\end{align}
$i=0,1,2,\dots$, and define $\bar{\mathcal{C}}_r = \bigcap_{i=0}^r \mathcal{C}_i$ with $\mathcal{C}_i = \{x: h_i(x) \ge 0\}$ and $\mathcal{C} = \{x: h(x) \ge 0\}$. Suppose that there exists $r$ such that, for any $x \in \bar{\mathcal{C}}_r$, we have $\frac{\partial{h_i}}{\partial x}B(x)u \ge 0$ for $i<r$ and \begin{equation}\label{eq:zcbf-const-thm1}
     \frac{\partial{h}_r}{\partial x}(g(x)+B(x)u)+\frac{1}{2}\textrm{tr}(\sigma^\top B^\top(x)(\frac{\partial^2 h_r}{\partial x^2})  B(x)\sigma) \ge -h_r(x).
\end{equation}
Then $P_r(x^t \in \mathcal{C} \hspace{2mm} \forall t) = 1$ if $x^0 \in \bar{\mathcal{C}}_r$.
\end{theorem}

\begin{proof}
Provided $u^t$ satisfies \eqref{eq:zcbf-const-thm1}, by Lemma \ref{lemma:state-invariant-zcbf}, we have $x^t \in {\mathcal{C}}_r \hspace{2mm} \forall t$ if $x^0 \in \mathcal{C}_r$, which is equivalent to $h_r(x^t) \ge 0 \hspace{2mm} \forall t$ if $x^0 \in \mathcal{C}_r$. By the construction of the high-degree ZCBFs as in \eqref{eq:construct-zcbf-high-deg}, we have \begin{equation}\label{eq:thm1_proof_zcbf_constru}
    h_r(x) = \frac{\partial{h_{r-1}}}{\partial x}g(x) + \frac{1}{2}\textrm{tr}(\sigma^\top B^\top(\frac{\partial^2 h_{r-1}}{\partial x^2}) B\sigma)+h_{r-1}(x) \ge 0.
\end{equation}
Since $\frac{\partial{h_{r-1}}}{\partial x}B(x)u \ge 0$, once added to both sides of \eqref{eq:thm1_proof_zcbf_constru}, we have \begin{align*}
    h_r(x)+ \frac{\partial{h_{r-1}}}{\partial x}B(x)u &= \frac{\partial{h_{r-1}}}{\partial x}(g(x)+B(x)u)\\&+\frac{1}{2}\textrm{tr}(\sigma^\top B^\top(\frac{\partial^2 h_{r-1}}{\partial x^2}) B\sigma)\\
    &+h_{r-1}(x) \ge 0,
\end{align*}
which satisfies \eqref{eq:zcbf-def},  and thus we have $x^t \in {\mathcal{C}}_{r-1} \hspace{2mm} \forall t$ if $x^0 \in \mathcal{C}_{r-1}$. By induction, we have \begin{equation*}
    x^t \in \mathcal{C}_i \hspace{2mm} \forall t \hspace{2mm}\textrm{if}\hspace{2mm} x^0 \in \mathcal{C}_i, i=0,1,2,\dots,r,
\end{equation*}
which is equivalent to $x^t \in \mathcal{C}_0 = \mathcal{C} \hspace{2mm} \forall t$, if $x^0 \in \bigcap_{i=0}^r \mathcal{C}_i = \bar{\mathcal{C}}_r$.
\end{proof}

\section{Control Compositionality}\label{sec3}
Considering the optimal control action takes a linear form  (\eqref{eq:sas-optimal-control-pf} in the single-agent scenario and \eqref{eq:mas-optimal-control-pf} in the multi-agent scenario), assume the component problems and the composite problem share certain dynamical information and satisfy certain conditions on the final cost; then a composite control action can be obtained analytically by combining existing component control actions.

Assume there are $F$ problems in MASs to solve, governed by identical dynamics \eqref{eq:cont_dyn_mas}, running cost rates \eqref{eq:mas-running-cost}, a set of interior joint states $\bar{\mathcal{I}}_i$ and boundary joint states $\bar{\mathcal{B}}_i$, but differ at the final costs. For factorial subsystem $\bar{\mathcal{N}}_i$, denote the final cost of the component problem $f$ on joint states $\bar{x}_i$ as $\phi_i^{\{f\}}(\bar{x}_i^{t_f})$, and the corresponding desirability function as $Z^{\{f\}}(\bar{x}_i^{t_f},t_f)$. Suppose the composite final cost $\phi_i(\bar{x}_i^{t_f})$ satisfies
  \begin{equation}\label{eq:mas-pf-final-cost-comp-condition}
      \phi_i(\bar{x}_i^{t_f}) = -\log(\sum_{f=1}^{F} \omega_i^{\{f\}} \exp(-\phi_i^{\{f\}}(\bar{x}_i^{t_f}))),
  \end{equation}
for some set of weights $\omega_i^{\{f\}}$. Then the composite desirability function can be computed on the boundary states $\bar{\mathcal{B}}_i$ as:\begin{equation}\label{eq:mas-pf-desirability-comp-condition}
    Z(\bar{x}_i^{t_f},t_f) = \sum_{f=1}^F \omega_i^{\{f\}}Z^{\{f\}}(\bar{x}_i^{t_f},t_f).
\end{equation}
Since the desirability function solves a linear HJB equation after the exponential  transformation \eqref{eq:mas-desir}, once the composite condition on the desirability function \eqref{eq:mas-pf-desirability-comp-condition} holds on the boundary, it holds everywhere, and the compositionality of optimal control actions can be extended to a task-generalization setting, i.e., the component problems and composite problem furthermore have different terminal states. The result of control generalization is formulated in the next theorem. Such formulation is especially useful when the component control actions are analytically solvable but costly; then the control action solving the new task can be constructed by composition in a sample-free manner and is less computationally expensive. 

\begin{theorem}
[Continuous-time MAS compositionality]\label{theorem_cont} Suppose there are $F$ multi-agent LSOC problems in continuous-time on factorial subsystem $\bar{\mathcal{N}}_i$ with joint states $\bar{x}_i$ and central agent state $x_i$, governed by the same joint dynamics \eqref{eq:cont_dyn_mas}, running cost rates \eqref{eq:mas-running-cost}, and the set of interior joint states $\overline{\mathcal{I}}_i$, but various terminal costs $\phi_i^{\{f\}}$ and terminal joint states $\bar{x}_i^{d^{\{f\}}}$.  Define the terminal joint state for a new problem as $\bar{x}_i^d$ and the composition weights as \begin{equation}\label{eq:continuous_weight}
    \bar{\omega}_i^{\{f\}} = \exp(-\frac{1}{2}(\bar{x}_i^{d}-\bar{x}_i^{d^{\{f\}}})^\top P(\bar{x}_i^{d}-\bar{x}_i^{d^{\{f\}}})),
\end{equation}
with $P$ being a positive definite diagonal matrix representing the kernel widths. Suppose the terminal cost for the new problem satisfies \begin{equation}\label{thm2_final_cost_composition}
    \phi_i(\bar{x}_{i}^{t_f}) = -\log(\sum\nolimits_{f=1}^F \tilde{\omega}_i^{\{f\}} \exp(-\phi_i^{\{f\}}(\bar{x}_{i}^{t_f}))),
\end{equation} where $\bar{x}_i^{t_f}$ denotes the boundary joint states and $\tilde{\omega}_i^{\{f\}} = \frac{\bar{\omega}_i^{\{f\}}}{\sum_{f=1}^{F}\bar{\omega}_i^{\{f\}}} $ can be interpreted as the probability weights. The optimal control law solving the new problem is obtained by a weighted combination of the existing controllers
\begin{equation}
    \bar{u}_i^{*}(\bar{x}_i, t) = \sum\nolimits_{f=1}^{F}\bar{W}_i^{\{f\}}(\bar{x}_i,t)\bar{u}_i^{*^{\{f\}}}(\bar{x}_i,t),
\end{equation}
with \begin{equation}
    \bar{W}_i^{\{f\}}(\bar{x}_i,t) = \frac{\tilde{\omega}_i^{\{f\}}Z^{\{f\}}(\bar{x}_i,t)}{\sum\nolimits_{e=1}^{F} \tilde{\omega}_i^{\{e\}}Z^{\{e\}}(\bar{x}_i,t)}.
\end{equation}
\end{theorem}
\begin{proof}
From composition of the terminal-cost function in continuous-time given by \eqref{thm2_final_cost_composition}, we have the following composition relationship of desirability functions by definition and it holds on the boundary joint states:   \begin{equation}\label{eq:thm2_des_relation_boundary}
    Z(\bar{x}_{i}^{t_f},t) = \sum\nolimits_{f = 1}^{F} \tilde{\omega}_i^{\{f\}}Z^{\{f\}}(\bar{x}_{i}^{t_f},t).
\end{equation} 
Since the desirability function solves a linear HJB equation, once condition \eqref{eq:thm2_des_relation_boundary} holds on the boundary, it holds everywhere and we have:
\begin{equation}\label{eq:thm2_des_relation_interior}
    Z(\bar{x}_{i},t) = \sum\nolimits_{f = 1}^{F} \tilde{\omega}_i^{\{f\}}Z^{\{f\}}(\bar{x}_{i},t).
\end{equation} 
For the composite problem, the joint optimal control action in the form of \eqref{eq:mas-optimal-control-pf}, can thus be reduced to
\begin{align*}
    &\quad \bar{u}_i^{*}(\bar{x}_i, t) = \bar{\sigma}_i \bar{\sigma}_i^\top \bar{B}_i^\top (\bar{x}_i) \frac{\nabla_{\bar{x}_i} Z(\bar{x}_i,t)}{Z(\bar{x}_i,t)}\\
    &= \bar{\sigma}_i \bar{\sigma}_i^\top \bar{B}_i^\top (\bar{x}_i) \frac{\nabla_{\bar{x}_i} \Big[{\sum\nolimits_{f=1}^{F} \tilde{\omega}_i^{\{f\}}Z^{\{f\}}(\bar{x}_i,t)}\Big]}{\sum\nolimits_{f=1}^{F} \tilde{\omega}_i^{\{f\}}Z^{\{f\}}(\bar{x}_i,t)}\\
    &= \frac{\sum\nolimits_{f=1}^{F}\bar{\sigma}_i \bar{\sigma}_i^\top \bar{B}_i^\top (\bar{x}_i)\nabla_{\bar{x}_i} \Big[\tilde{\omega}_i^{\{f\}}Z^{\{f\}}(\bar{x}_i,t)\Big]}{\sum\nolimits_{e=1}^{F} \tilde{\omega}_i^{\{e\}}Z^{\{e\}}(\bar{x}_i,t)}\\
    &= \frac{\sum\nolimits_{f=1}^{F}\bar{\sigma}_i \bar{\sigma}_i^\top \bar{B}_i^\top (\bar{x}_i)Z^{\{f\}}(\bar{x}_i,t)\nabla_{\bar{x}_i} \Big[\tilde{\omega}_i^{\{f\}}Z^{\{f\}}(\bar{x}_i,t)\Big]}{\sum\nolimits_{e=1}^{F} \tilde{\omega}_i^{\{e\}}Z^{\{e\}}(\bar{x}_i,t)Z^{\{f\}}(\bar{x}_i,t)}\\
    &= \sum\limits_{f=1}^{F}{\frac{\tilde{\omega}_i^{\{f\}}Z^{\{f\}}(\bar{x}_i,t)}{\sum\limits_{e=1}^{F} \tilde{\omega}_i^{\{e\}}Z^{\{e\}}(\bar{x}_i,t)}\bar{\sigma}_i \bar{\sigma}_i^\top \bar{B}_i^\top (\bar{x}_i)\frac{\nabla_{\bar{x}_i} Z^{\{f\}}(\bar{x}_i,t)}{Z^{\{f\}}(\bar{x}_i,t)}}\\
    &= \sum\nolimits_{f=1}^{F}\bar{W}_i^{\{f\}}(\bar{x}_i,t)\bar{u}_i^{*^{\{f\}}}(\bar{x}_i,t),
\end{align*}
with \begin{equation*}
    \bar{W}_i^{\{f\}}(\bar{x}_i,t) = \frac{\tilde{\omega}_i^{\{f\}}Z^{\{f\}}(\bar{x}_i,t)}{\sum\nolimits_{e=1}^{F} \tilde{\omega}_i^{\{e\}}Z^{\{e\}}(\bar{x}_i,t)}. 
\end{equation*}
\end{proof}

\section{Composition on Certified-safe Control Actions}\label{sec4}
We first propose a revised optimal stochastic control law incorporating ZCBF in a multi-agent setting, which can provide safety guarantees by enforcing certain conditions.
\begin{theorem}
[Safe and optimal control in MASs] \label{coro:mas-safe-opt-ctrl} In continuous-time joint dynamics \eqref{eq:cont_dyn_mas}, for factorial subsystem $\bar{\mathcal{N}}_i$ with joint states $\bar{x}_i$ and central agent state $x_i$, the desirability function is represented by $Z(\bar{x}_i,t) = \exp[-V_i(\bar{x}_i,t)/\lambda_i]$ with $\lambda_i \in \mathbb{R}$ and $V_i(\cdot,t)$ being the value function, and the joint optimal control action is given by \begin{equation}\label{eq:mas_opt_ctrl_law_joint}
    \bar{u}_i^*(\bar{x}_i,t) = \bar{\sigma}_i\bar{\sigma}_i^\top \bar{B}_i^\top(\bar{x}_i)\frac{\nabla_{\bar{x}_i} Z(\bar{x}_i,t)}{Z(\bar{x}_i,t)}.
\end{equation}
As discussed in Section \ref{pf:mas-setup}, only the central agent $i$ of each factorial subsystem $\bar{\mathcal{N}}_i$ selects or samples its local control action  $u_i^*(\bar{x}_i,t)$ from \eqref{eq:mas_opt_ctrl_law_joint}. Define the ZCBFs for high-degree multi-agent systems as $h_0(x_i) = h(x_i)$, and subsequently \begin{align}
    h_{k+1}(x_i) &= \frac{\partial{h_k}}{\partial x_i}g_i(x_i)\nonumber\\ &+ \frac{1}{2}\textrm{tr}(\sigma_i^\top B_i^\top(\frac{\partial^2 h_k}{\partial {x_i}^2}) B_i\sigma_i)+h_k(x_i), 
\end{align}
$k=0,1,2,\dots$, and define $\bar{\mathcal{C}}_{r,i} = \bigcap_{k=0}^r \mathcal{C}_{k,i}$ with $\mathcal{C}_{k,i} = \{x_i: h_k(x_i) \ge 0\}$ and $\mathcal{C}_i = \{x_i: h(x_i) \ge 0\}$. The CBF constraint for the high-degree multi-agent system is given by \begin{equation}\label{eq:cbf-high-deg-const}
    \frac{\partial{h}_r}{\partial x_i}(g_i(x_i)+B_i(x_i)u)+\frac{1}{2}\textrm{tr}(\sigma_i^\top B_i^\top\frac{\partial^2 h_r}{\partial x_i^2}B_i\sigma_i) \ge -h_r(x_i),
\end{equation}
where $r$ is selected such that $ \frac{\partial{h}_k}{\partial x_i}B_i(x_i)u \ge 0$ holds for $k<r$ and \eqref{eq:cbf-high-deg-const} holds for $k=r$. Then, the optimal control for the central agent $i$ of each factorial subsystem $\bar{\mathcal{N}}_i$ that can also guarantee the subsystem safety (state-invariant) within $\mathcal{C}_i$ is given by 
\begin{align}\label{eq:component-safe-opt-control}
    {u}_{s,i}^* &= \argmin_{u \in \mathbb{R}^P} \|{u}_i^* - u\|_2, \\
     \textrm{s.t.} &\eqref{eq:cbf-high-deg-const}  \nonumber
\end{align}
where $u_i^*$ is the optimal local control action for the central agent $i$ sampled from \eqref{eq:mas_opt_ctrl_law_joint}.
\end{theorem}
\begin{proof}
The first part (optimal control law) of the theorem comes from the solution to a linearly solvable optimal control problem and is established in \eqref{eq:mas-optimal-control-pf} of Section \ref{pf:mas-setup}. However, simply applying the optimal control strategy may not render the system safe considering the system stochastic noise, since the control goal such as obstacle-avoidance is captured in the cost function in the form of soft constraints. For the second part, a constrained optimization framework is applied to both minimizing the difference towards the ideal optimal control sampled from \eqref{eq:mas-optimal-control-pf} and enforcing the ZCBF constraints \eqref{eq:cbf-high-deg-const}  for safety concerns. For high-degree systems, $\frac{\partial{h_k}}{\partial x_i}B_i(x_i) = 0$ for some $x_i$ when $k<r$; and $r$ is selected as the least positive integer such that $\frac{\partial{h_r}}{\partial x_i}B_i(x_i) \ne 0$ and where $u$ starts to explicitly show up linearly in \eqref{eq:cbf-high-deg-const}. Finally, if the condition \eqref{eq:cbf-high-deg-const} holds, the system state is invariant within $\mathcal{C}_i$ and guaranteed by Theorem \ref{lemma:zcbf-high-deg}.
\end{proof}
Furthermore, utilizing the linear compositionality of the optimal control action in multi-agent systems, along with the fact that the ZCBF constraints are affine in control, the above results can be extended to a task-generalization setting, where the optimal solution to a new control problem can be achieved by taking a weighted mixture of existing control actions. Meanwhile, the composite control action is certified-safe given by an additional step of post-composite constrained optimization using ZCBFs.
\begin{theorem}[Generalization of safe and optimal control in MASs] \label{coro:generalizetion-mas-safe-opt}
Suppose there are $F$ multi-agent LSOC problems in continuous-time on factorial subsystem $\bar{\mathcal{N}}_i$ with joint states $\bar{x}_i$ and central agent state $x_i$, governed by the same joint dynamics \eqref{eq:cont_dyn_mas}, running cost rates \eqref{eq:mas-running-cost}, and the set of interior joint states $\bar{\mathcal{I}}_i$, but various terminal costs $\phi_i^{\{f\}}$ and terminal joint states $\bar{x}_i^{d^{\{f\}}}$. The safe and optimal control action $u_{s,i}^{*^{\{f\}}} (f=1,2,\dots,F)$ for central agent $i$  solving each single problem is computed via \eqref{eq:component-safe-opt-control}. Denote the terminal joint state for a new problem as $\bar{x}_i^d$ and define the composition weights as \begin{equation}
    \bar{\omega}_i^{\{f\}} = \exp(-\frac{1}{2}(\bar{x}_i^d - \bar{x}_i^{d^{\{f\}}})^\top P(\bar{x}_i^d - \bar{x}_i^{d^{\{f\}}})),
\end{equation}
with $P$ being a positive definite diagonal matrix representing the kernel widths. Suppose the terminal cost for the new problem satisfies \begin{equation}\label{thm2_final_cost_composition}
    \phi_i(\bar{x}_{i}^{t_f}) = -\log(\sum\nolimits_{f=1}^F \tilde{\omega}_i^{\{f\}} \exp(-\phi_i^{\{f\}}(\bar{x}_{i}^{t_f}))),
\end{equation} where $\bar{x}_i^{t_f}$ denotes the boundary joint states and $\tilde{\omega}_i^{\{f\}} = \frac{\bar{\omega}_i^{\{f\}}}{\sum_{f=1}^{F}\bar{\omega}_i^{\{f\}}} $. The optimal control that solves the new problem is directly computable through a weighted combination of the component problem control solutions:\begin{equation}\label{eq:composite-control-nocbf-mas}
    u_i^*(\bar{x}_i,t) = \sum_{f=1}^F \bar{W}_i^{\{f\}}(\bar{x}_i,t)u_{s,i}^{*^{\{f\}}}(\bar{x}_i,t)
\end{equation}
with \begin{align*}
    \bar{W}_i^{\{f\}}(\bar{x}_i,t) &= \frac{\tilde{\omega}_i^{\{f\}}Z_i^{\{f\}}(\bar{x}_i,t)}{\sum_{e=1}^{F}\tilde{\omega}_i^{\{e\}}Z_i^{\{e\}}(\bar{x}_i,t)},
\end{align*} where the local control action $u_{s,j}^{*^{\{f\}}}$ is only sampled or selected from the computed joint control $\bar{u}_{s,j}^{*^{\{f\}}}$ on factorial subsystem $\bar{\mathcal{N}}_j$. Furthermore, the optimal control for central agent $i$ of each factorial subsystem $\bar{\mathcal{N}}_i$ that can solve the composite task and also guarantee the subsystem safety (state-invariant) within a desired safe set $\mathcal{C}_i = \{x_i:h(x_i) \ge 0\}$ is given by
\begin{align}\label{eq:composite-safe-opt-control}
    {u}_{s,i}^* &= \argmin_{u \in \mathbb{R}^P} \|{u}_i^* - u\|_2, \\
     \textrm{s.t.} &\eqref{eq:cbf-high-deg-const}  \nonumber
\end{align}
where $u_i^*$ is computed using \eqref{eq:composite-control-nocbf-mas}.
\end{theorem}
\begin{proof} The first part of the theorem is an immediate result from Theorem \ref{theorem_cont}, where the safe control from Theorem \ref{coro:mas-safe-opt-ctrl} is applied as the primitives, and only the sampled central agent control is considered. Since  \eqref{eq:cbf-high-deg-const} is affine in control $u$ on considered system dynamics, and each component problem control law solves a constrained optimization problem where the baseline optimal control is linearly solvable (details for the optimal control linear compositionality can also be found in \cite{song2021compositionality}),  the linearity property preserves and enables the compositionality of such safe and optimal control actions computed using Theorem \ref{coro:mas-safe-opt-ctrl}. Furthermore, condition \eqref{eq:cbf-high-deg-const} applied to the composite control action from \eqref{eq:composite-control-nocbf-mas} ensures that the system is state-invariant within $\mathcal{C}_i$ by Theorem \ref{lemma:zcbf-high-deg}.
\end{proof}

\section{Simulation Results}
\label{sec5}
We performed numerical simulations in Matlab for both single-agent systems (single UAV) and cooperative networked multi-agent systems (cooperative UAV team) to demonstrate the proposed results. Each UAV is described by the following continuous-time dynamics:
\begin{equation}\label{eq:sim-uav-dyn}
	\setlength\arraycolsep{1pt}
	{\scalefont{0.8} \left(\begin{matrix}
		dx_i\\
		dy_i\\
		dv_i\\
		d\varphi_i
		\end{matrix}\right) = 
		\left(\begin{matrix}
		v_i \cos \varphi_i\\
		v_i \sin \varphi_i\\
		0\\
		0
		\end{matrix}\right) dt +  \begin{pmatrix}
		0 & 0\\
		0 & 0\\
		1 & 0\\
		0 & 1
		\end{pmatrix} \left[ \left( \begin{matrix}
		u_i\\
		w_i
		\end{matrix}  \right) dt + \begin{pmatrix}
		\sigma_i & 0\\
		0 & \nu_i
		\end{pmatrix}d\omega_i
		\right], }
\end{equation}
where $(x_i,y_i), v_i, \varphi_i$ denote the position coordinate, forward velocity and heading angle for UAV $i$. The system state vector is $(x_i,y_i,v_i,\varphi_i)^\top$; the forward acceleration $u_i$ and angular velocity $w_i$ are the control inputs, and $\omega_i$ is the standard Brownian-motion disturbance. We specify the noise level as $\sigma_i = 0.05$ and $\nu_i = 0.025$ throughout the simulation.

In all the following experiments, we compute the baseline optimal control actions (\eqref{eq:sas-optimal-control-pf} in the single-agent scenario and \eqref{eq:mas_opt_ctrl_law_joint}, \eqref{eq:composite-control-nocbf-mas} in the multi-agent scenario) in a path integral approximation framework, and obtain the constrained optimal control actions using the above theorems. In the following  obstacle-avoidance tasks, each individual obstacle is described by an independent set of ZCBF constraints. We also explicitly check that the start point satisfies all the ZCBF constraints in place (especially $x^0 \in \bar{\mathcal{C}}_r$ as  in Theorem \ref{lemma:zcbf-high-deg}), which renders the proposed methodology applicable.
\subsection{Single-agent system experiments with the proposed control strategy}
\subsubsection{Single-problem experiment}
We first compare the performance between the filtered control actions with CBFs and the baseline optimal control actions in a simple case. The UAV is governed by the dynamics in \eqref{eq:sim-uav-dyn} and is tasked to fly from the start $(5,5)$ towards a target $(35,20)$ while avoiding the obstacles. The running cost for the UAV is in the following form:\begin{equation}\label{sim:single-agent-single-prob-running-cost}
    q(\mathbf{x}) = \|(x,y)-(x^{t_f},y^{t_f})\|_2 - d^{\max},
\end{equation}
where $\|(x,y)-(x^{t_f},y^{t_f})\|_2$ computes the distance to the goal position for the UAV, and $d^{\max}$ denotes the distance between the initial position and target position for the UAV. The obstacles are incorporated by greater state-related cost values (i.e. $q(x) = 160$) in the baseline optimal control action computation. Here, we utilize CBFs for enhanced safety and construct a set of CBFs for the description of one circular obstacle centered at $(x_c,y_c)$ with a radius of $r_c$ as follows:\begin{align}
    h_0(x) &= h(x) = (x-x_c)^2+(y-y_c)^2 -(r_c+D_s)^2, \label{sim:eq-cbf-h0}\\
    h_1(x) &= \frac{\partial{h_0}}{\partial x}g(x) + \frac{1}{2}\textrm{tr}(\sigma^\top B^\top (\frac{\partial^2 h_0}{\partial x^2})B \sigma)+h_0(x) \nonumber\\
    &= 2(x-x_c)v\cos\phi + 2(y-y_c)v\sin\phi\nonumber \\
    &+(x-x_c)^2+(y-y_c)^2-(r_c+D_s)^2, \label{sim:eq-cbf-h1}
\end{align}
where $g(x), B, \sigma$ are defined in the dynamics \eqref{eq:cont_dyn} and $D_s$ is a user-defined safety margin.

The simulation runs for 10 times independently, under both the safe optimal control actions  and the ideal optimal control actions, respectively. The execution trajectories are shown in Fig. \ref{fig_comp_with_vs_without_cbf}, where the obstacles are denoted by the filled circles, the trajectories under the safe optimal control are denoted by the green lines, and the trajectories under the ideal optimal control are denoted by the red lines. As Fig. \ref{fig_comp_with_vs_without_cbf} shows, the performance of the ideal optimal control is not always uniform on avoiding the lowest obstacle and relies on fine parameter-tuning. However, the trajectories using the safe optimal control actions can always guarantee a safe margin surrounding the obstacles throughout all the simulations. Meanwhile, the safety margin can be tuned as a hyper-parameter for achieving the 
least conservative and feasible trajectory.
\begin{figure}[!t]
\centerline{\includegraphics[width=\columnwidth]{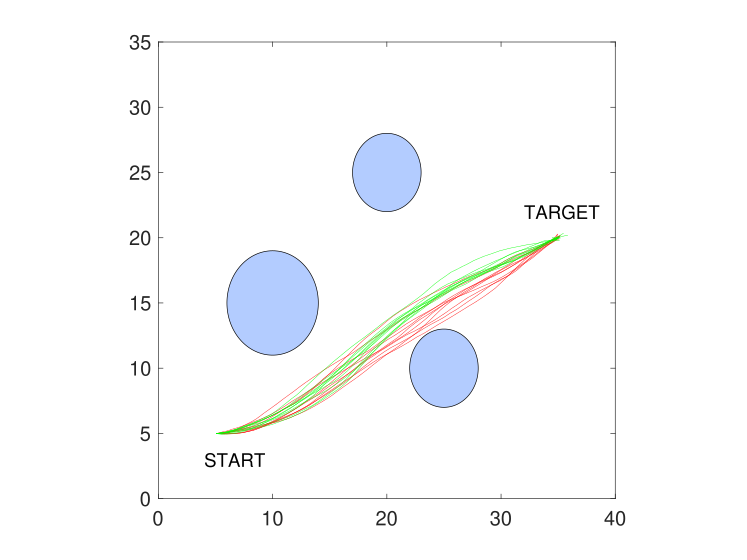}}
\caption{Comparison between executive trajectories with (denoted in the green lines) and without (denoted in the red lines) CBF constraints in a single-agent case.}
\label{fig_comp_with_vs_without_cbf}
\end{figure}

\begin{figure}[!ht]
\centerline{\includegraphics[width=\columnwidth]{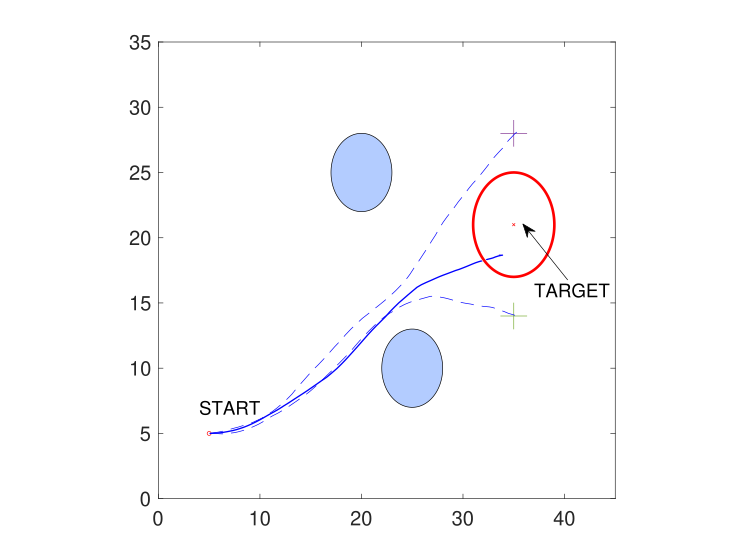}}
\caption{Composition in a single-agent case, where the dashed lines denote the component problem execution trajectories, and the solid lines denote the composite problem execution trajectory. The red circle centered at the target denotes the allowable error range.}
\label{fig_comp_single_agent}
\end{figure}
\subsubsection{Task-generalization experiment}
In this example, we consider two component problems seeking safe optimal control actions under the identical dynamics \eqref{eq:sim-uav-dyn}, running cost rates \eqref{sim:single-agent-single-prob-running-cost}, set of interior states, but have different final costs and terminal states. The final costs take the form of $\phi = \frac{d}{2}(|x-x_d|+c)+\alpha$, where the two problems have different cost parameters $c,d,\alpha$. The first problem ensures that the UAV starting from $(5,5)$ reaches the upper target $(35,28)$, while the second problem ensures that the UAV starting from $(5,5)$ reaches the lower target $(35,14)$. Further, we consider how these  safe control actions can help with solving a new problem aiming at reaching the target $(35,21)$. Here, we take a weighted mixture of safe and optimal control actions solving the component problems and obtain the optimal solution constrained by the ZCBFs. The execution trajectories for the component problems and the composite problem are shown in Fig. \ref{fig_comp_single_agent}. 


In Fig. \ref{fig_comp_single_agent}, the executive trajectories of the two component problems are denoted by the dashed lines, and the executive trajectory running the composite safe optimal control action is denoted by the solid line. The targets of the two component tasks are denoted by the cross markers. As Fig. \ref{fig_comp_single_agent} shows, all these trajectories can avoid the obstacles with sufficient safety margins. The component problem solution can lead the UAV to the targeted goal accurately. However, the safe optimal control action by composition leads the UAV to a target within some acceptable error range, denoted by the red solid circle.

\subsection{Evaluation of the safe and optimal control law for a single task in a networked MAS under various safety margins }

For the networked MAS simulations, we consider a cooperative UAV team as illustrated in Fig. \ref{fig_team_flight}, where UAVs 1 and 2 fly cooperatively (distance-minimized) and UAV 3 flies independently towards the goal joint  states while avoiding the obstacles. According to the factorization introduced in Section \ref{pf:mas-setup}, the joint states of the three factorial subsystems are $\bar{\mathbf{x}}_1 = [\mathbf{x}_1;\mathbf{x}_2]^\top,\bar{\mathbf{x}}_2 = [\mathbf{x}_1;\mathbf{x}_2;\mathbf{x}_3]^\top, \bar{\mathbf{x}}_3 = [\mathbf{x}_2;\mathbf{x}_3]^\top$, where $\mathbf{x}_i = [x_i;y_i;v_i;\varphi_i]^\top$ and $(x_i,y_i), v_i, \varphi_i$ denote the position coordinate, forward velocity and heading angle for UAV $i$. The system joint dynamics can be described by \eqref{eq:cont_dyn_mas}.  The coordination between UAVs is considered by the running cost in the following form:\begin{align}
    q_1(\bar{\mathbf{x}}_1) &= 0.7\cdot (\|(x_1,y_1)-(x_1^{t_f},y_1^{t_f})\|_2-d_1^{\textrm{max}}) \\ 
    & \hspace{35pt}+ 1.4 \cdot(\|(x_1,y_1)-(x_2,y_2)\|_2-d_{12}^{\textrm{max}}),\nonumber\\
    q_2(\bar{\mathbf{x}}_2)   &= 0.7\cdot(\|(x_2,y_2)-(x_2^{t_f},y_2^{t_f})\|_2-d_2^{\textrm{max}}) \\
    &\hspace{35pt} + 1.4\cdot(\|(x_2,y_2)-(x_1,y_1)\|_2-d_{21}^{\textrm{max}}),\nonumber\\
    q_3(\bar{\mathbf{x}}_3) &= \|(x_3,y_3)-(x_3^{t_f},y_3^{t_f})\|_2-d_3^{\textrm{max}},
\end{align}
where $\|(x_i,y_i)-(x_i^{t_f},y_i^{t_f})\|_2$ calculates the distance to the goal position for UAV $i$, $\|(x_i,y_i) - (x_j,y_j)\|_2$ calculates the distance between UAVs $i$ and $j$, $d_i^{\max}$ denotes the distance between the initial position and target position for UAV $i$, and $d_{ij}^{\max}$ denotes the initial distance between UAVs $i$ and $j$. The cost parameters and coefficients can be tuned for better performance and algorithm stability.

\begin{figure}[!ht]
\centerline{\includegraphics[width=0.9\columnwidth]{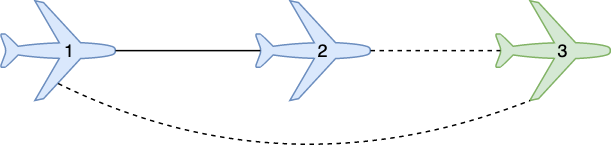}}
\caption{A cooperative UAV team with UAVs 1 and 2 flying cooperatively and UAV 3 flying independently.}
\label{fig_team_flight}
\end{figure}

We first evaluate optimization-based control actions from \eqref{eq:mas_opt_ctrl_law_joint} on the networked MAS, and run 10 independent simulations. The trajectories for the UAV team are shown in Fig. \ref{fig_component_without_cbf}, where trajectories of UAVs 1, 2 and 3 are denoted by the red, blue, and green lines, respectively. The two obstacles are represented by the filled circles. The start point for UAV $2$ is labeled by `S2' in Fig. \ref{fig_component_without_cbf}. As the figure shows, all UAVs can reach the appointed targets. However, although the control performance can be improved by tuning the obstacle state costs and running cost coefficients, no guarantees on the obstacle-avoiding goal is achievable. Among the observed runs, there are some cases when UAV 3 collides with the lower obstacle due to the stochastic noise.  
\begin{figure}[!t]
\centerline{\includegraphics[width=\columnwidth]{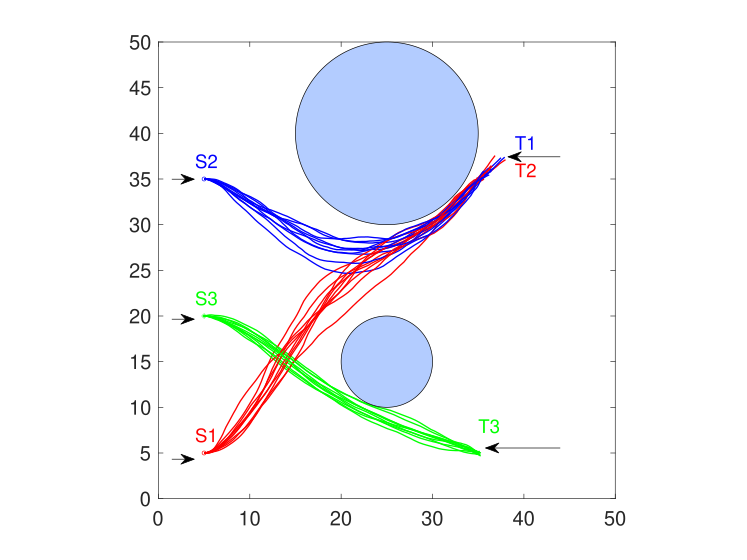}}
\caption{Illustration of one  problem solution on a networked UAV team without CBF constraints. The three agents' trajectories are denoted in the red, blue, and green lines, respectively. The `S1' marker denotes the start of UAV 1 and the `T1' marker denotes the target of UAV 1.}
\label{fig_component_without_cbf}
\end{figure}

We further extend the experiments using the control actions  subject to the CBF constraints according to \eqref{eq:component-safe-opt-control}. The CBFs are designed similarly as \eqref{sim:eq-cbf-h0} and \eqref{sim:eq-cbf-h1}. We run 10 independent simulations and the executive trajectories are shown in Fig. \ref{fig_component_with_cbf}. Similar to the case under ideal optimal control actions, all UAVs can reach the target states, and the coordination between UAVs 1 and 2 is achieved. Furthermore, although every single trajectory may differ much due to the stochastic noise, all trajectories of each single UAV can avoid the placed obstacles with some safety margins. 

\begin{figure}[!t]
\centerline{\includegraphics[width=\columnwidth]{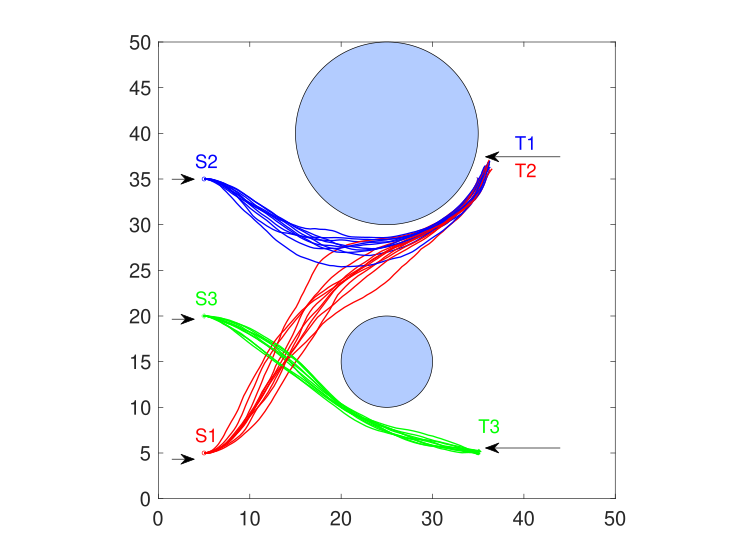}}
\caption{Illustration of one problem solution on a networked UAV team with CBF constraints. The three agents' trajectories are denoted in the red, blue, and green lines, respectively. The `S1' marker denotes the start of UAV 1 and the `T1' marker denotes the target of UAV 1.}
\label{fig_component_with_cbf}
\end{figure}

Compared with the case of using ideal optimal control actions without CBFs, where no margin can be achieved surrounding the obstacles for the trajectories, the safety margin achieved from Theorem \ref{coro:mas-safe-opt-ctrl} can further be tuned. We evaluate the relationship between the commanded margins in CBF design (e.g. $D_s$ parameter in \eqref{sim:eq-cbf-h1}) and the achieved minimum distance to the nearest obstacles for all UAVs throughout the simulations, and a quantitative illustration result is given in Fig. \ref{fig_component_quant_result}. We can observe that the UAVs running unfiltered optimal control actions fail to meet the safety margin, where UAVs under safe optimal control actions can avoid all the obstacles and the minimum distance to both obstacles is larger than the threshold set in the CBF design in \eqref{sim:eq-cbf-h0} and \eqref{sim:eq-cbf-h1}.

\subsection{Performance of the composite safe and optimal control law generalizing to a new task in a networked MAS}

\begin{figure}[!ht]
\centerline{\includegraphics[width=0.9\columnwidth]{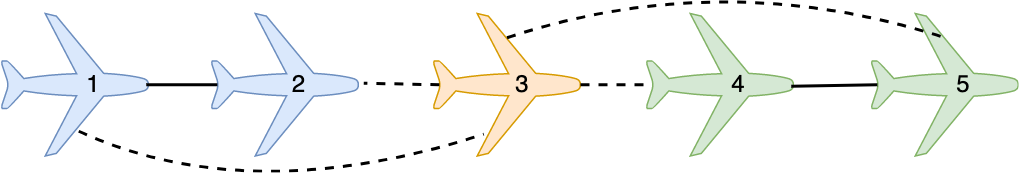}}
\caption{A cooperative UAV team with UAVs 1 and 2, 4 and 5 flying cooperatively and UAV 3 flying independently.}
\label{fig_team_flight_five}
\end{figure}
We further evaluate the proposed safe optimal control strategy in a task-generalization setting on the cooperative UAV team in Fig. \ref{fig_team_flight_five}. The five UAVs work in three groups,  where UAVs 1 and 2, 4 and 5 fly cooperatively (distance-minimized) and UAV 3 flies independently towards the goal while avoiding some obstacles. We consider two component problems, subject to identical joint dynamics \eqref{eq:cont_dyn_mas}, joint running costs \eqref{eq:mas-running-cost}, and set of interior joint states $\bar{\mathcal{I}}_i$ for factorial subsystem $\bar{\mathcal{N}}_i$, but different final costs and terminal joint states. In the two problems, the target position for all the UAVs are $(35,28)$ and $(35,14)$, respectively. In each problem, the safe optimal control leading the UAV team to the target is obtained according to \eqref{eq:component-safe-opt-control}, and the execution trajectories are demonstrated in Fig. \ref{fig_component_problem_comp}, where the trajectories of UAVs 1,2, 3, 4 and 5 are denoted by the red, blue, green, magenta and cyan lines, respectively. The target positions in the two different problems are labeled by the stars.

\begin{figure}[!t]
\centerline{\includegraphics[width=0.9\columnwidth]{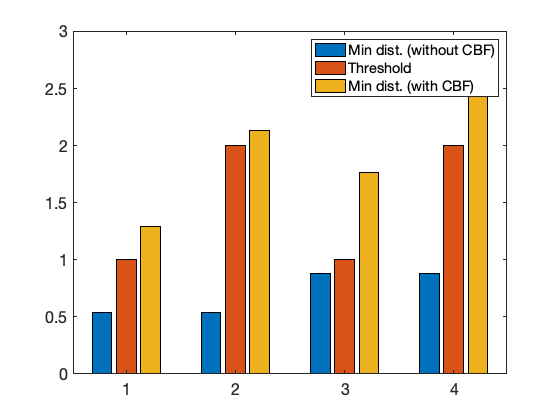}}
\caption{Minimum distance to the two obstacles under different CBF margins throughout the simulations (data batches 1 and 2 show the minimum distance to the upper obstacle, data batches 3 and 4 show the minimum distance to the lower obstacle).}
\label{fig_component_quant_result}
\end{figure}

\begin{figure}[!t]
\centerline{\includegraphics[width=\columnwidth]{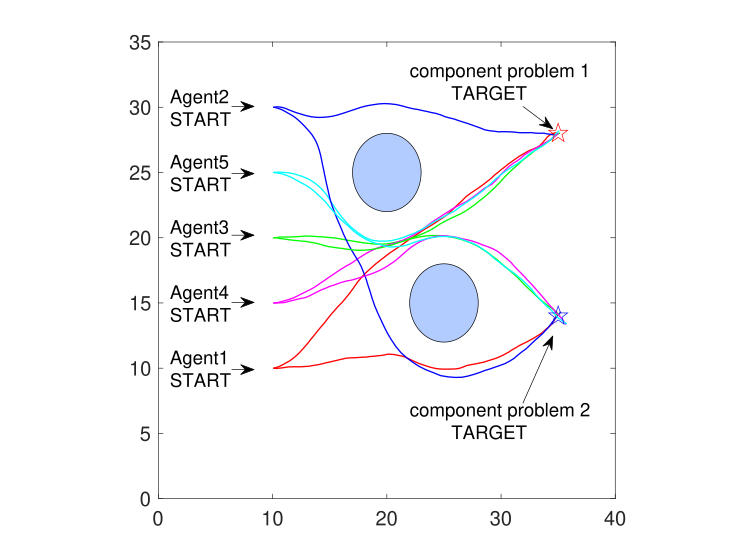}}
\caption{Illustration of the solutions to the component problems for composition on the safe optimal control law, with lines in red, blue, green, magenta and cyan denoting the trajectories of agents 1, 2, 3, 4 and 5, respectively. Target of each component problem is denoted by a star.}
\label{fig_component_problem_comp}
\end{figure}

\begin{figure}[!t]
\centerline{\includegraphics[width=\columnwidth]{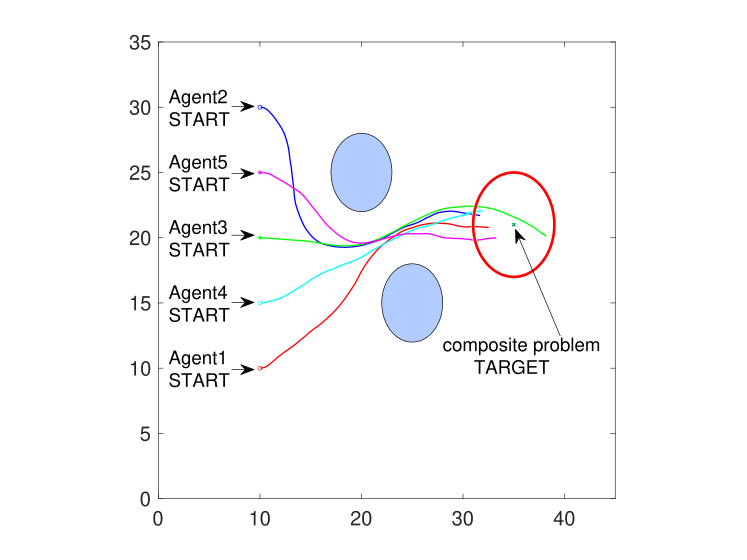}}
\caption{Illustration of the execution trajectories using the safe composite control law, with lines in red, blue, green, magenta, and cyan denoting the trajectories of agents 1, 2, 3, 4 and 5, respectively. The red circle centered at the new target denotes the allowable error range.}
\label{fig_composite_problem_comp}
\end{figure}
Once the component problem safe optimal control action is obtained, a weighted mixture on the primitives specified by \eqref{eq:composite-control-nocbf-mas} can be taken to achieve the composite safe optimal control action, and a further constrained optimization step according to \eqref{eq:composite-safe-opt-control} will ensure that the resulting optimal control action is safe-guaranteed. The execution trajectories of the UAV team using the composite safe control action to solve a new problem are shown in Fig. \ref{fig_composite_problem_comp}, where the red solid circle demonstrates an allowable error range. As Fig. \ref{fig_composite_problem_comp} shows, under the filtered composite safe optimal control action, all the UAVs can avoid the obstacles with suitable safe margins, and UAVs 1 and 2, 4 and 5 can cooperate. They can also reach the target position but subject to some error.

Similar to what Fig. \ref{fig_comp_single_agent} illustrates, beyond the impact of the stochastic noise, the composition is not exactly accurate here since the desirability function solving the transformed linear HJB equation discussed in Section \ref{pf:soc-problems} is time-dependent for continuous-time systems, and the composition on time is not achievable. Furthermore, though each component controller by itself is in the feedback-control form, the composite control is not performing feedback on the new continuous states and thus fails to have the state error-compensating capacity; it is also fragile to the external noise. However, by running several simulations and selecting the optimal local control action for each agent independently, the obtained performance, as illustrated in Fig. \ref{fig_composite_problem_comp}, is satisfying, and the terminal error is controllable. Also, the composite safe control action proposed to solve the new problem using Theorem \ref{coro:generalizetion-mas-safe-opt} is obtained in a sample-free manner by taking a weighted mixture of primitives and getting filtered by the CBF constraints. It  is worth to apply especially in the case when each component problem solution can be solved analytically but is expensive to compute, when consideration of effort in solving a new problem dominates the   control precision.

\section{Conclusion}
\label{sec6}
In this paper, we developed a framework of safe generalization of optimal control utilizing control barrier functions (CBFs) and linearly-solvable property in stochastic system control, in both single-agent and cooperative networked multi-agent system cases. The proposed control action simultaneously ensures optimality and guarantees safety by enforcing the CBF constraints, while minimizing the difference away from the ideal optimal control. Considering the linearity of considered CBF constraints and compositionality of linear-solvable optimal control (LSOC), we further extend the safe optimal control framework to a task generalization setting, where a weighted mixture of computed actions for component problems is taken to solve a new problem. The safety of such composite control action is incorporated by additional CBF constraints as a filter after the composition. The composite safe optimal control action is obtained in a sample-free manner and thus is less computationally-expensive, while the safety property can be reserved by the additional CBF constraint. We evaluate the proposed approach on numerical simulations of a single UAV and two cooperative UAV teams with obstacle-avoidance and target-reaching goals. The constructed composition control law can drive the teamed UAV to a new target within some acceptable error range. The error can be explained by the system stochastic noise and the error in composing a time-dependent desirability function of continuous-time dynamics. We hope the proposed strategy can be applied in scenarios when the computational cost of the component task safe optimal control solution dominates the  precision of the implemented composite control. Future work will consider the safety guarantees on the control action generalization capability on networked MASs under incomplete information where not every state information is measurable, and such scenario can simulate the real-world environment with high-fidelity. Another direction of research for future work consists of leveraging images and high-dimensional sensor data, and introducing perception-based closed-loop control to the current framework.

\bibliography{reference}
\bibliographystyle{IEEEtran}

\end{document}